\documentclass[1p,number]{elsarticle}

\usepackage{amsfonts,amssymb}

\usepackage[all,dvips,color,arrow,curve]{xy}
\xyoption{arc}
\xyoption{poly}

\newtheorem{theorem}{Theorem}
\newtheorem{lemma}[theorem]{Lemma}
\newtheorem{proposition}[theorem]{Proposition}
\newtheorem{corollary}[theorem]{Corollary}

\renewcommand{\qed}{$\Box$}
\newenvironment{proof}
{\noindent{\bf Proof.}\ }
{\hfill\qed\par\bigskip}

\begin{document}

\title{Dichotomy for tree-structured trigraph list homomorphism problems}

\author[tomas]{Tom\'as Feder}
\ead{tomas@theory.stanford.edu}

\author[sfu]{Pavol Hell\fnref{nserc}}
\ead{pavol@cs.sfu.ca}

\author[sfu]{David G. Schell}
\ead{dschell@cs.sfu.ca}

\author[liafa]{Juraj Stacho}
\ead{jstacho@liafa.jussieu.fr}

\fntext[nserc]{This research was partially supported by P. Hell's NSERC Discovery Grant.}

\address[tomas]{268 Waverley St., Palo Alto, CA 94301, USA}

\address[sfu]{School of Computing Science, Simon Fraser University\\ 8888
University Drive, Burnaby, B.C., Canada V5A 1S6}

\address[liafa]{LIAFA -- CNRS and Universit\'e Paris Diderot -- Paris VII,\\
Case 7014, 75205 Paris Cedex 13, France }

\begin{abstract}
Trigraph list homomorphism problems (also known as list matrix partition
problems) have generated recent interest, partly because there are concrete
problems that are not known to be polynomial time solvable or $NP$-complete.
Thus while {\em digraph} list homomorphism problems enjoy {\em dichotomy} 
(each problem is $NP$-complete or polynomial time solvable), such dichotomy 
is not necessarily expected for trigraph list homomorphism problems. However,
in this paper, we identify a large class of trigraphs for which list homomorphism 
problems do exhibit a dichotomy. They consist of trigraphs with a tree-like 
structure, and, in particular, include all trigraphs whose underlying graphs 
are trees. In fact, we show that for these tree-like trigraphs, the trigraph list 
homomorphism problem is polynomially equivalent to a related {\em digraph} list
homomorphism problem.  We also describe a few examples illustrating that our
conditions defining tree-like trigraphs are not unnatural, as relaxing them may 
lead to harder problems.
\end{abstract}

\maketitle

\section{Introduction}
A {\em trigraph} $H$ consists of a set $V=V(H)$ of vertices, and two  disjoint
sets of directed edges on $V$ -- the set of {\em weak edges} $W(H)  \subseteq V
\times V$, and the set of {\em strong edges} $S(H) \subseteq V \times  V$. If
both edge sets $W(H), S(H)$, viewed as relations on $V$, are  symmetric, we 
have a {\em symmetric}, or {\em undirected} trigraph. A weak, respectively  strong,
edge $vv$ is called a {\em weak}, respectively {\em strong, loop} at $v$.

The {\em adjacency matrix} of a trigraph $H$, with respect to an  enumeration
$v_1, v_2, \dots, v_n$ of its vertices, is the $n \times n$ matrix $M $ over $0,
1, *$, in which $M_{i,j} = 0$ if $v_iv_j$ is not an edge, $M_{i,j} = * $ if
$v_iv_j$ is a weak edge, and $M_{i,j} = 1$ if $v_iv_j$ is a strong  edge. Note
that a trigraph $H$ is symmetric if and only if its adjacency matrix is symmetric.

We consider the class of digraphs included in the class of trigraphs, by viewing
each digraph $H$ as a trigraph with the same vertex set $V(H)$, and with the 
weak edge set $W(H)=E(H)$ and strong edge set $S(H)= \emptyset$. Conversely, 
if $H$ is a trigraph, the {\em associated digraph} of $H$  is the digraph with the 
same vertex set $V(H)$, and with the edge set $E(H) = W(H) \cup S(H)$. 
Moreover, the {\em underlying graph} of the trigraph $H$ is the underlying 
graph of the associated digraph, and the {\em symmetric graph} of the trigraph
$H$ is the symmetric graph of the associated digraph. To be specific, $xy$ is
an edge of the {\em underlying graph} of $H$ just if $xy \in W(H) \cup S(H)$ or
$yx \in W(H) \cup S(H)$, and $xy$ is an edge of the {\em symmetric graph} of 
$H$ just if $xy \in W(H) \cup S(H)$ and $yx \in W(H) \cup S(H)$. These conventions 
allow us to extend the usual graph and digraph terminology to trigraphs. We speak, 
for instance, of {\em adjacent vertices, components, neighbours, cutpoints, or 
bridges} of a trigraph $H$, meaning the corresponding notions in the associated 
digraph of $H$, or in its underlying graph; and we speak of {\em symmetric edges,
symmetric neighbours, etc.} in a trigraph $H$, meaning the edges, neighbours,
etc., in the symmetric graph of $H$.

Let $G$ be a digraph and $H$ a trigraph. A {\em homomorphism} of $G$  to $H$ is
a mapping $f : V(G) \rightarrow V(H)$ such that  the following two conditions
are satisfied for each $u\neq v$:
\begin{enumerate}[--]
\item if $uv \in E(G)$ then $f(u)f(v) \in W(H) \cup S(H)$
\item if $uv \not\in E(G)$ then $f(u)f(v) \not\in S(H)$.
\end{enumerate}

In other words, edges of $G$ must map to either weak or strong edges  of $H$,
and non-edges of $G$ must map to either non-edges or weak edges of $H$.

If each vertex $v$ of the digraph $G$ has a {\em list} $L(v) \subseteq V(H)$, 
then a {\em list homomorphism} of $G$ to $H$, with respect to the lists $L$, 
is a homomorphism $f$ of $G$ to $H$ such that $f(v) \in L(v)$ for all $v \in V(G)$.
Following standard practice~\cite{hombook}, we also call a homomorphism of $G$
to $H$ an {\em $H$-colouring of $G$}, and a list homomorphism of $G$ to $H$
(with respect to the lists $L$) a {\em list $H$-colouring} of $G$ (with respect to $L$).

Suppose $H$ is a fixed trigraph. The {\em $H$-colouring problem} \mbox{HOM}$(H)$
has as instances digraphs $G$, and asks whether or not $G$ admits an
$H$-colouring. The {\em list $H$-colouring problem} \mbox{L-HOM}$(H)$ has as
instances digraphs $G$ with lists $L$, and asks whether or not $G$ admits a list
$H$-colouring with respect to $L$.  As noted earlier, $H$ could be a digraph,
viewed as a trigraph (with $W(H)=E(H)$, $S(H)=\emptyset$). Digraph homomorphism
and list homomorphism problems have been of much interest \cite{hombook}. If
necessary, we will emphasize the distinction between {\em trigraph} list
homomorphism problems and {\em digraph} list homomorphism problems, depending on
whether $H$ is a trigraph or a digraph respectively, i.e., whether $H$ has any
strong edges or not.  However, note that the input $G$ is always~a~digraph.

For a fixed trigraph $H$, the list $H$-colouring problem L-HOM$(H)$ concerns the
existence of vertex partitions of the input digraphs $G$. For instance, if $H$
is the undirected trigraph with $V(H)=\{0,1\}$, with a strong loop at $1$ and a
weak edge joining $0$ and $1$, then an $H$-colouring of $G$ is  precisely a
partition of $V(G)$ into a clique and an independent set. Thus $G$ is
$H$-colourable if and only if $G$ is a split graph. Many graph partition
problems, especially those arising from the theory of perfect graphs, can be
formulated as trigraph homomorphism (or list homomorphism) problems; this is
discussed in detail in \cite{fhkm}. Equivalently, all these problems can be
described in terms of the adjacency matrix of the trigraph, in the language of
matrix partitions and list partitions (see \cite{cehm,fhkm,fhks,kim}).  In this
paper it will be more convenient to emphasize the trigraph (rather than the
matrix) terminology, since we are dealing with the structure of the trigraph
$H$.

It is generally believed \cite{federvardi} that for each digraph $H$ the
$H$-colouring problem HOM$(H)$ is $NP$-complete or polynomial time solvable.
(This is equivalent to the so-called {\em CSP Dichotomy Conjecture} of Feder and
Vardi \cite{federvardi}.) One special case when the dichotomy conjecture is
known to hold is the case of undirected graphs (i.e., symmetric digraphs).  In
this case, HOM$(H)$ is polynomial time solvable if $H$ has a loop or is
bipartite and is $NP$-complete otherwise \cite{homdich}. For the list
homomorphism problem \mbox{L-HOM$(H)$}, it is shown in \cite{fhh} that
L-HOM$(H)$ is polynomial time solvable if $H$ is a so-called {\em bi-arc graph}
(a simultaneous generalization of reflexive interval graphs and bipartite graphs
whose complements are circular arc graphs), and is $NP$-complete otherwise. A
more general result of Bulatov \cite{bulatov} handles all constraint
satisfaction problems, implying, in particular, that for each digraph $H$, the
list $H$-colouring problem L-HOM$(H)$ is $NP$-complete or polynomial time
solvable. By contrast, this is not known for trigraph list homomorphism
problems, and in \cite{fulcsp} it is only proved that for each trigraph $H$, the
list $H$-colouring problem is $NP$-complete or {\em quasi-polynomial} (of
complexity $n^{O(\log^kn)}$).  All list $H$-colouring problems L-HOM$(H)$ for
trigraphs $H$ with three or fewer vertices have been classified as $NP$-complete
or polynomial time solvable in \cite{kim}.  For symmetric trigraphs with four
vertices, this has been accomplished in \cite{cehm}, with the exception of a
single trigraph $H$; the corresponding problem remains open and has earned the
name {\em the stubborn problem} (Figure \ref{fig:ex1}a). The best known algorithm
for this problem has complexity $n^{O(\log n / \log \log n)}$, implying it is
unlikely to be $NP$-complete \cite{fhks}. Thus (polynomial / $NP$-complete)
dichotomy for trigraph list homomorphism problems seems less likely than for
{\em digraph} list homomorphism problems.

In this paper we prove dichotomy for the class of {\em trigraph trees}, i.e.,
for trigraphs whose underlying graph is a tree. It turns out that if $H$ is a
trigraph tree, then the list $H$-colouring problem is polynomially equivalent to
a list $H^-$-colouring problem where $H^-$ is a {\em digraph} obtained from $H$
by removing all vertices with a strong loop and removing all other strong edges
$xy$ (and their converses $yx$ if any).

We conduct the proof of this result in such a way that it in fact implies the
dichotomy for a large class of trigraphs, which includes all digraphs and 
all trigraph trees. We think of these trigraphs as {\em tree-like}, although
it is only the structure of the strong edges that is tree-like. For trigraphs 
that are not tree-like (in our definition), we illustrate the possible 
complications. We believe that our class of tree-like trigraphs covers 
an important portion of the class of trigraphs $H$ in which the strong 
edges do not significantly impact the complexity of the list $H$-colouring 
problem.

\section{Tools}

Let $H$ be a trigraph, and let $G$ be a digraph with list $L(u)\subseteq V(H)$
for each $u\in V(G)$.  We shall denote by $n$ the number of vertices in $G$, 
and by $k$ the number of vertices in $H$. We say that lists $L'$ are {\em a
reduction of $L$}, if we have $L'(u)\subseteq L(u)$ for each $u\in V(G)$.

In the following text, we shall say that lists $L$ can be {\em reduced} to
satisfy property $P$ to mean that for every instance consisting of a graph 
$G$ with lists $L$, there exist lists $L'$ on $G$
such that 

\begin{enumerate}[(i)]
\item
the lists $L'$ can be found in polynomial time
\item
the lists $L'$ are a reduction of $L$
\item
the lists $L'$ satisfy property $P$, and 
\item
$G$ has a list $H$-colouring with respect to $L$ if and only if $G$
has a list $H$-colouring with respect to $L'$.
\end{enumerate}

We shall say that lists $L$ can be {\em transformed} to satisfy property $P$
to mean that for every instance $G$ with lists $L$, there exists a set $\mathcal
L=\{L_i\}_{i\in I}$ of lists such that 

\begin{enumerate}[(i)]
\item
the set $\mathcal L$ can be constructed in polynomial time
\item
each $L_i$ is a reduction of $L$
\item
each $L_i$ satisfies property $P$, and 
\item
$G$ has a list $H$-colouring with respect to $L$ if and only if 
there exists an $i \in I$ such that $G$ has a list $H$-colouring with respect 
to $L_i$.
\end{enumerate}

Note that while reducing the lists results in a single problem, transforming the
lists produces a family of problems, as illustrated in Figure \ref{fig:1}.

\begin{figure}[h]
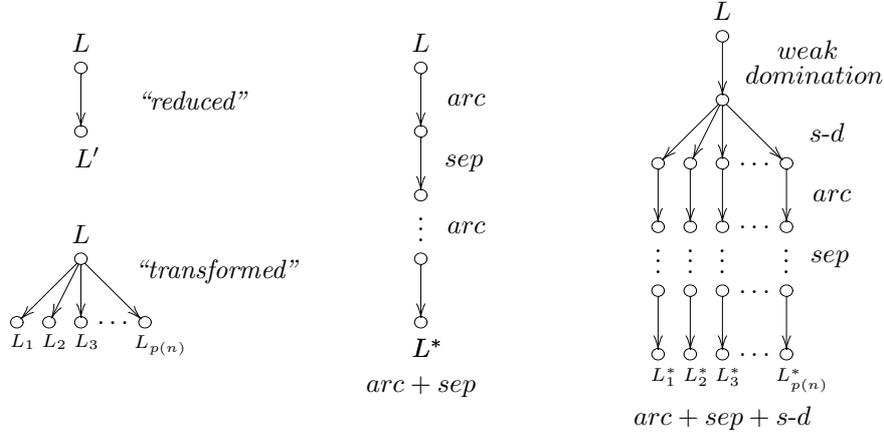

\centering
$
\xy/r2pc/:
(-1,0.5)*[o][F]{\phantom{s}}="a1";
(-1,-0.5)*[o][F]{\phantom{s}}="a2";
{\ar "a1";"a2"};
"a1"+(0,0.4)*{L};
"a2"+(0.1,-0.4)*{L'};
(0,0)*[r]{\mbox{\em ``reduced''}};
(-1,-2.5)*[o][F]{\phantom{s}}="a1";
(-2,-3.5)*[o][F]{\phantom{s}}="a2";
(-1.5,-3.5)*[o][F]{\phantom{s}}="a3";
(-1,-3.5)*[o][F]{\phantom{s}}="a4";
(-0.5,-3.5)*[o]{\ldots}="a5";
(0,-3.5)*[o][F]{\phantom{s}}="a6";
{\ar "a1";"a2"};
{\ar "a1";"a3"};
{\ar "a1";"a4"};
{\ar "a1";"a6"};
"a1"+(0,0.4)*{L};
"a2"+(0.1,-0.3)*{_{L_1}};
"a3"+(0.1,-0.3)*{_{L_2}};
"a4"+(0.1,-0.3)*{_{L_3}};
"a6"+(0.25,-0.35)*{_{L_{p(n)}}};
(0,-2.7)*[r]{\mbox{\em ``transformed''}};
(4.3,0.5)*[o][F]{\phantom{s}}="a1";
(4.3,-0.5)*[o][F]{\phantom{s}}="a2";
(4.3,-1.5)*[o][F]{\phantom{s}}="a3";
(4.3,-1.9)*[o]{\vdots};
(4.3,-2.5)*[o][F]{\phantom{s}}="a4";
(4.3,-3.5)*[o][F]{\phantom{s}}="a5";
{\ar "a1";"a2"};
{\ar "a2";"a3"};
{\ar "a4";"a5"};
"a1"+(0,0.4)*{L};
"a5"+(0.1,-0.4)*{L^*};
(4.8,0)*[r]{arc};
(4.8,-1)*[r]{sep};
(4.8,-2)*[r]{arc};
(4.3,-4.5)*{arc+sep};
(9,1)*[o][F]{\phantom{s}}="a0";
(9,0)*[o][F]{\phantom{s}}="a1";
(8,-1)*[o][F]{\phantom{s}}="a21";
(8.5,-1)*[o][F]{\phantom{s}}="a22";
(9,-1)*[o][F]{\phantom{s}}="a23";
(9.5,-1)*[o]{\ldots}="a24";
(10,-1)*[o][F]{\phantom{s}}="a25";
{\ar "a0";"a1"};
{\ar "a1";"a21"};
{\ar "a1";"a22"};
{\ar "a1";"a23"};
{\ar "a1";"a25"};
(8,-2)*[o][F]{\phantom{s}}="a31";
(8.5,-2)*[o][F]{\phantom{s}}="a32";
(9,-2)*[o][F]{\phantom{s}}="a33";
(9.5,-2)*[o]{\ldots}="a34";
(10,-2)*[o][F]{\phantom{s}}="a35";
{\ar "a21";"a31"};
{\ar "a22";"a32"};
{\ar "a23";"a33"};
{\ar "a25";"a35"};
(8,-2.4)*[o]{\vdots};
(8.5,-2.4)*[o]{\vdots};
(9,-2.4)*[o]{\vdots};
(10,-2.4)*[o]{\vdots};
(8,-3)*[o][F]{\phantom{s}}="a41";
(8.5,-3)*[o][F]{\phantom{s}}="a42";
(9,-3)*[o][F]{\phantom{s}}="a43";
(9.5,-3)*[o]{\ldots}="a44";
(10,-3)*[o][F]{\phantom{s}}="a45";
(8,-4)*[o][F]{\phantom{s}}="a51";
(8.5,-4)*[o][F]{\phantom{s}}="a52";
(9,-4)*[o][F]{\phantom{s}}="a53";
(9.5,-4)*[o]{\ldots}="a54";
(10,-4)*[o][F]{\phantom{s}}="a55";
{\ar "a41";"a51"};
{\ar "a42";"a52"};
{\ar "a43";"a53"};
{\ar "a45";"a55"};
"a0"+(0,0.4)*{L};
"a5"+(0.1,-0.4)*{L^*};
(10,0.8)*[r]{weak};
(9.5,0.4)*[r]{domination};
(10.5,-0.5)*[r]{s\mbox{-}d};
(10.5,-1.5)*[r]{arc};
(10.5,-2.5)*[r]{sep};
"a51"+(0.1,-0.35)*{_{L^*_1}};
"a52"+(0.1,-0.35)*{_{L^*_2}};
"a53"+(0.1,-0.35)*{_{L^*_3}};
"a55"+(0.25,-0.4)*{_{L^*_{p(n)}}};
(9,-5)*{arc+sep+s\mbox{-}d};
\endxy
$
\caption{Illustrating the concepts of ``reduced'' and ``transformed''. The
circles represent lists. The middle and the right figure illustrate Lemma \ref{lem:5}
and \ref{lem:8} respectively. Labels {\em arc, sep,} and {\em s-d} are
shorthands for arc-consistency,  separator-consistency, and
sparse-dense-consistency.\label{fig:1}}
\end{figure}

\medskip

We say that lists $L$ are {\em non-empty} if the list $L(u)$ for each vertex
$u\in V(G)$ is not empty. Clearly, if $G$ admits a list homomorphism to $H$ with
respect to $L$, then the lists $L$ are non-empty. Hence, to avoid trivial cases
in what follows, we shall always assume that lists $L$~are~non-empty.

We say that lists $L$ are {\em arc-consistent} if for each $u,v\in V(G)$ and
each $x\in L(u)$, there exists $y\in L(v)$ such that (i) $xy\in W(H)\cup S(H)$
if $uv\in E(G)$, (ii) $yx\in W(H)\cup S(H)$ if $vu\in E(G)$, (iii) $xy\not\in
S(H)$ if $uv\not\in E(G)$, and (iv) $yx\not\in S(H)$ if $vu\not\in E(G)$.

\begin{lemma}\label{lem:1}
Lists $L$ can be reduced to be arc-consistent.
\end{lemma}
\begin{proof}
If some $u,v\in V(G)$ violate the above condition for $x\in L(u)$, no
$H$-colouring of $G$ with respect to $L$ maps $u$ to $x$. Hence, $x$ can
be removed from $L(u)$ without changing the existence of solutions. We 
repeatedly test for such violations and reduce the lists if a violation is found. 
After at most $n\times k$ such steps either we obtain arc-consistent lists.
(Note that if one final list is empty then all are empty, by arc-consistency.)
\end{proof}

We say that lists $L$ contain {\em representatives}, if there is a set
$X\subseteq V(H)$ such that
\begin{enumerate}[(i)]
\item for each $v\in V(G)$, the list $L(v)\subseteq X$, and
\item for each $x\in X$, there is a vertex $v\in V(G)$ with $L(v)=\{x\}$,
\end{enumerate}

\begin{lemma}\label{lem:2}
Lists $L$ can be transformed to contain representatives.
\end{lemma}
\begin{proof}
We construct the set $\mathcal L=\{L_i\}_{i\in I}$ as follows.  The lists $L_i$
are obtained in performing in every possible way the following process. For each
vertex $x \in V(H)$, we either select a vertex of $G$ and set its list to
$\{x\}$, or we remove $x$ from all lists. The resulting lists $L_i$ contain
representatives for the set $X$ of vertices $x$ for which some list in $G$ was
set to $\{x\}$.  Moreover, the vertices of $H - X$ do not appear in $L_i$.
\end{proof}

The {\em symmetric trigraph $H'$ associated with $H$} is the trigraph on the
vertices of $H$ with strong edges $xy$ iff $xy,yx\in S(H)$, and with weak edges
$xy$ iff $xy\not\in S(H)$ or $yx\not\in S(H)$ but $xy,yx\in W(H)\cup
S(H)$.

A graph $G$ is {\em chordal} if $G$ contains no induced cycle of length four or
longer. A {\em chordal completion} $G'$ of $G$ is a chordal graph on the
vertices of $G$ with $E(G')\supseteq E(G)$. A chordal completion $G'$ is {\em
minimal} if no chordal completion $G''$ of $G$ satisfies $E(G'')\subsetneqq
E(G')$.

\begin{proposition}\cite{tarjan}\label{prop:tarjan}
Let $G'$ be a minimal chordal completion of $G$, and let $C$ be a clique of $G$.
Then there are no edges in $G'$ between vertices of different components of
$G-C$.
\end{proposition}

We say that lists $L$ contain {\em strong representatives}, if for each strong
loop $x$ of $H$, there is a set $S_x\subseteq V(H)$ such that

\begin{enumerate}[(i)]
\item $L(v)\subseteq S_x$ whenever $x\in L(v)$,
\item each vertex of $S_x\setminus\{x\}$ is a symmetric neighbour of $x$, and
\item $S_x\setminus \{x\}$ is connected in the symmetric graph of~$H$.
\end{enumerate}

\begin{lemma}\label{lem:3}
Lists $L$ can be transformed to contain strong representatives.
\end{lemma}
\begin{proof}
Let $G$ be a digraph with lists $L$, and let $f$ be a list $H$-colouring of $G$
with respect to $L$. By Lemmata \ref{lem:1} and \ref{lem:2}, we may assume that
the lists $L$ are arc-consistent and contain representatives.

Let $x$ be a strong loop of $H$.  Since $f$ is a homomorphism, if $f(u)=f(v)=x$
for $u,v\in V(G)$, we must have both $uv$ and $vu$ in $E(G)$.  Hence, the set
$C=f^{-1}(x)$ induces a symmetric clique in $G$.  

Let $G'$ be the symmetric graph of $G$, and $H'$ be the symmetric trigraph
associated with $H$. Clearly, $C$ also induces a clique in $G'$.  Moreover, it
is easy to show that $f$ is also a list $H'$-colouring of $G'$ with respect to
$L$.

Now, suppose that $x$ appears on some list. Then, since the lists $L$ contain
representatives, there is a vertex $v_x\in V(G)$ with $L(v_x)=\{x\}$.

Let $N$ denote the subset of $V(H)$ containing $x$ and its symmetric neighbours.
Let $M$ denote the subset of $V(G)$ containing $v_x$ and its symmetric
neighbours. Let $B$ denote the subset of $V(G)$ containing all vertices $u\in
V(G)$ with $x\in L(u)$.  Since the lists $L$ are arc-consistent, we must have
$B\subseteq M$. Also, for each $u\in M$, we have $L(u)\subseteq N$.

Now, since $f$ respects the lists $L$,  we have $x\in L(u)$, for each $u\in C$.
Hence, $C\subseteq B$, and therefore, $C$ induces a clique in $G'[B]$.
Furthermore, since $B\subseteq M$, we have $L(u)\subseteq N$ for each $u\in B$.
Hence, $f$ restricted to $B$ is a homomorphism of $G'[B]$ to $H'[N]$.  In
particular, each component of $G'[B]-C$ maps by $f$ to a unique component of
$H'[N]-x$.

Now, let $G''$ be a minimal chordal completion of $G'[B]$. Let $C''$ be the
vertices of a maximal clique of $G''$ that completely contains $C$.  By
Proposition \ref{prop:tarjan}, all vertices of $C''\setminus C$ belong to one
component of $G'[B]-C$. Hence, by the above remark, there is a unique component
$K$ of $H'[N]-x$ such that $f(u)\in V(K)$ for each $u\in C''\setminus C$.

It now follows that $G$ admits a list $H$-colouring with respect to $L$ if and
only if for some maximal clique $C''$ of a minimal chordal completion $G''$ of
$G'[B]$, and some component $K$ of $H'[N]-x$, the graph $G$ admits a list
$H$-colouring with respect to $L$ such that the vertices outside of $C''$ do not
map to $x$, and the vertices in $C''$ map to $x$ or the vertices of $K$.  Hence,
we can modify the lists $L$ by removing $x$ from the vertices outside $C''$ and
by reducing the lists of the vertices in $C''$ to $V(K)\cup\{x\}$. Such lists,
clearly, contain a strong representative for $x$. (Take $S_x$ to
be~$V(K)\cup\{x\}$.)

To conclude, we remark that the graph $G'$ and the sets $B$ and $N$, as well as,
a minimal chordal completion $G''$ of $G'[B]$ by the result \cite{tarjan} can be
found in polynomial time. Since $G''$ is chordal, it has at most $n$ maximal
cliques \cite{fulkersongross}. Also, there are at most $k$ different components
$K$ of $H'[N]-x$. Hence, we can reduce the problem to at most $(nk)^k$ different
instances with strong representatives. The proof is now complete.
\end{proof}

Note that if lists $L$ contain representatives, respectively strong
representatives, then any non-empty reduction $L'$ of $L$ also contains
representatives, respectively strong representatives. 
\medskip

\def\setminuss{\hspace{0.15em}\backslash\hspace{-0.3em}\backslash\hspace{0.1em}}

Let $F$ be a set of edges of $H$.  Let $G\setminuss F$ denote the graph obtained
from $G$ by removing all edges $uv\in E(G)$ such that $xy\in F$ for some $x\in
L(u)$~and~$y\in L(v)$.

We say that lists $L$ are {\em separator-consistent on $F$}, if for each
component $C$ of $G\setminuss F$ and each component $K$ of $H\setminus F$,
\begin{enumerate}[(i)]
\item either there exists a list $K$-colouring of $C$ with respect to $L$, 
\item or no vertex of $K$ appears on the list of any vertex of $C$.
\end{enumerate}

\begin{lemma}\label{lem:4}
If L-HOM$(H\setminus F)$ is polynomial time solvable, then the lists $L$ can be
reduced to be separator-consistent on $F$.
\end{lemma}
\begin{proof}
We obtain a separator-consistent lists $L'$ from $L$ as follows. For each
component $C$ of $G\setminuss F$, and each component $K$ of $H\setminus F$, we
test if $C$ admits a list $K$-colouring with respect to $L$.  If not, then we
remove all vertices of $K$ from the lists of the vertices of $C$.  Since
homomorphisms map connected graphs only to connected graphs,
the claim follows.
\end{proof}

In particular, we have the following property.

\begin{lemma}\label{lem:5}
If L-HOM$(H\setminus F)$ is polynomial time solvable, then the lists $L$ can be
reduced to be arc-consistent and separator-consistent on $F$.
\end{lemma}
\begin{proof}
We apply Lemmata \ref{lem:1} and \ref{lem:4} to $L$ until the lists no longer
change. Since at each step the lists are reduced, the claim follows.
\end{proof}

We remark that for $X\subseteq V(H)$, we say that the lists $L$ are {\em
separator-consistent on $X$}, if they are separator-consistent on $F$, where $F$
is the set of edges of $H$ with at least one endpoint in $X$. \bigskip

Let $X$ and $Y$ be two sets of vertices of $H$ such that each vertex of $X$ has
a strong loop, and no vertex of $Y$ has a loop.

We say that lists $L$ are {\em sparse-dense-consistent on $X$ and $Y$} if for
each $v\in V(G)$ such that $L(v)\subseteq X\cup Y$, we have $L(v)\subseteq X$ or
$L(v)\subseteq Y$.

\begin{lemma}\label{lem:6}
If L-HOM$(H[X])$ and L-HOM$(H[Y])$ are polynomial time solvable, then the lists $L$
can be transformed to be sparse-dense-consistent on $X$ and $Y$.
\end{lemma}
\begin{proof}
Let $\mathcal S$ and $\mathcal D$ be classes of digraphs closed under taking
induced subgraphs. Suppose that there is a constant $c=c(\mathcal S,\mathcal D)$
such that each digraph in $\mathcal S\cap \mathcal D$ has at most $c$ vertices.
Then by \cite{fhkm}, for any $n$-vertex digraph $G$, there are at most $n^{2c}$
partitions $V(G)=V_1\cup V_2$ such that $G[V_1]\in\mathcal S$ and
$G[V_2]\in\mathcal D$. (We call each such partition an {\em $(\mathcal S,\mathcal
D)$-partition} of $G$.) All these partitions can be enumerated in time
$n^{2c+2}T(n)$ where $T(n)$ is the complexity of recognizing digraphs in
$\mathcal S$ and in $\mathcal D$.

Now, let $\mathcal S$ be the set of all digraphs admitting an $H[X]$-colouring
and $\mathcal D$ be the set of all digraphs admitting an $H[Y]$-colouring; we
have $c(\mathcal S,\mathcal D)\leq |X|\cdot |Y|$.  Let $Z$ denote the vertices
$u$ of $G$ with $L(u)\subseteq X\cup Y$. We observe that any list $H$-colouring
of $G$ with respect to $L$ induces an $(\mathcal S,\mathcal D)$-partition of
$G[Z]$. Hence, for each $(\mathcal S,\mathcal D)$-partition $Z=A\cup B$ of
$G[Z]$, we construct lists $L'$ on $H$ as follows: $L'(u)=L(u)\cap X$ if $u\in
A$, $L'(u)=L(u)\cap Y$ if $u\in B$, and $L'(u)=L(u)$ otherwise. Clearly, $L'$ is
sparse-dense consistent on $X$ and $Y$. Now, since there are at most $n^{k^2}$
such partitions, the claim follows.
\end{proof}

Let $D$ be a set of vertices of $H$, and let $x,y$ be two vertices of $H$.  We
say that $y$ {\em weakly dominates $x$ on $D$}, if for each $z\in D$, we have
that $zy$ is a weak edge in $H$ whenever $zx$ is a weak edge in $H$, and $yz$ is
a weak edge in $H$ whenever $xz$ is a weak edge in $H$.

We say that {\em $y$ weakly dominates $x$} if $y$ weakly dominates $x$ on
$V(H)$.  

\begin{lemma}\label{lem:7}
If $y$ weakly dominates $x$, then the lists $L$ can be reduced to satisfy $x\not\in
L(u)$ whenever $y\in L(u)$.
\end{lemma}
\begin{proof}
By Lemmata \ref{lem:1} and \ref{lem:2}, we may assume that the lists $L$ are
arc-consistent and contain representatives. Let $f$ be a list $H$-colouring of
$G$ with respect to $L$. Let $f'$ be a mapping such that $f'(u)=y$ if $f(u)=x$
and $y\in L(u)$, and $f'(u)=f(u)$ otherwise. We show that $f'$ is also a
homomorphism and the claim will follow.

Let $t$ be any vertex (including $x$ and $y$) of $H$. Suppose that $xt\in S(H)$,
$yt$ is not an edge, and $t$ appears on some list. We claim that no list $L(u)$
contains both $x$ and $y$. To prove this, let $v_t$ be a vertex with
$L(v_t)=\{t\}$. Suppose that $x,y\in L(u)$. If $uv_t\in E(G)$, then by
arc-consistency of $L$, we obtain $y\not\in L(u)$. If $uv_t\not\in E(G)$,
similarly $x\not\in L(u)$, a contradiction.

Moreover, by symmetry, if $tx\in S(H)$ and $ty$ is not an edge, or if $x$ and
$y$ exchange places, we also have that no list contains both $x$ and $y$.
In~addition, since $y$ weakly dominates $x$, we obtain that either both $x$
and~$y$ have strong loops, or both have no loops, or $y$ has a weak loop.

Now, it is not difficult to directly verify that $f'$ is a homomorphism using
the above observations and the fact that $y$ weakly dominates $x$.

Hence, $G$ has a list $H$-colouring with respect to $L$ if and only if $G$ has a
list $H$-colouring with respect to the lists $L'$ obtained from $L$ by removing $x$
from each list $L(v)$ that also contains $y$.  That concludes the proof.
\end{proof}

Let $X$, $Y$, and $Z$ be three sets of vertices of $H$ such that each vertex of
$X$ has a strong loop, no vertex of $Y$ has a loop, and each vertex of $Z$ has a
weak loop. Suppose that each vertex of $Z$ weakly dominates each vertex of $Y$,
and that we have $L(u)\subseteq X\cup Y\cup Z$ whenever $L(u)\cap
X\neq\emptyset$.

\begin{lemma}\label{lem:8}
Let $X,Y,Z$ be as above. If L-HOM$(H[X])$ and \mbox{L-HOM$(H-X)$} are polynomial
time solvable, then the lists $L$ can be transformed to be arc-consistent,
separator-consistent on $X$, and sparse-dense-consistent on $X$~and~$Y$.
\end{lemma}
\begin{proof}
We observe that by Lemma \ref{lem:7}, we may assume that no list $L(u)$ contains
both a vertex of $Y$ and a vertex of $Z$. Hence, either $L(u)\subseteq X\cup Z$,
or $L(u)\subseteq X\cup Y$, or $L(u)\cap X=\emptyset$ for each $u\in V(G)$. 

Now, let $L_i$ be one of the lists we obtain from $L$ by applying sparse-dense
consistency on $X$ and $Y$ (Lemma \ref{lem:6}). Then we must have
either $L_i(u)\subseteq X$, or $L_i(u)\subseteq Y$, or $L_i(u)\cap X=\emptyset$,
or $L_i(u)\cap Y=\emptyset$ for each $u\in V(G)$. In particular, any reduction
of $L_i$ must also satisfy this condition.  Hence, we can apply Lemma
\ref{lem:5} to $L_i$, and the claim follows.
\end{proof}

Let $A\cup B$ be a partition of the vertices of $H$.  Let $F$ be the edges of
$H$ that have exactly one endpoint in $A$. Let $X$ respectively $Y$ be the
vertices of $A$ respectively $B$ with at least one incident edge in $F$.

We say that $G$ is {\em separable} on $F$, if for each vertex $v\in V(G)$,
\begin{enumerate}[(i)]
\item $L(v)$ contains at most one vertex of $X$ and at most one vertex of $Y$,
\item if $L(v)$ contains a vertex of $X$, it contains no vertex of
$A\setminus X$, and
\item if $L(v)$ contains a vertex of $Y$, it contains no vertex of
$B\setminus Y$.
\end{enumerate}

\begin{lemma}\label{lem:9}
If L-HOM($H-X-Y$) is polynomial time solvable, then L-HOM($H$) is polynomial
time solvable on the class of all digraphs $G$ separable on $F$.
\end{lemma}
\begin{proof}
First, we observe if $G$ with lists $L$ is separable on $F$, then for any
reduction $L'$ of $L$, the graph $G$ with lists $L'$ is separable on $F$.
Hence, by Lemmata \ref{lem:1} and \ref{lem:2}, we may assume that the
lists $L$ are arc-consistent and contain representatives.

Next, we show that L-HOM$(H[A])$ can be solved in polynomial time on any
$G\setminuss F$ given that $G$ is separable on $F$. Indeed, consider a component
$C$ of $G\setminuss F$, and let $L'$ be lists such that $L'(u)=L(u)\cap A$ for
each $u\in V(C)$. Since $G$ is separable on $F$, we have either $L'(u)\subseteq
A\setminus X$, or $|L'(u)|=1$. Hence, if $B$ denotes the vertices with
$|L'(u)|=1$, then arc-consistency of $L$ implies that $C$ has a list
$H[A]$-colouring with respect to $L$ if and only if $C-B$ has a list
$(H-X-Y)$-colouring with respect to $L'$. Similarly, for L-HOM$(H[B])$. Hence,
L-HOM$(H\setminus F)$ can be solved in polynomial time for $G\setminuss F$ given
$G$ is separable~on~$F$.

It now follows by Lemma \ref{lem:5} that we may assume that the lists $L$ are
arc-consistent and separator-consistent on $F$. We also assume that lists $L$
are non-empty, since otherwise there is no solution.\medskip

Let $H_0$ be the trigraph constructed from $H[X\cup Y]$ by adding vertices $a$ and
$b$ with weak loops such that $a$ has a weak symmetric edge to each vertex of
$X$, and $b$ has a weak symmetric edge to each vertex~of~$Y$.

Let $L_0$ be the lists obtained from $L$ by replacing by $a$ each $z\in L(u)$ such
that $z\in A\setminus X$, and replacing by $b$ each $z\in L(u)$ such that $z\in
B\setminus Y$.  Also, let $A_0=X\cup\{a\}$ and $B_0=Y\cup \{b\}$.

Let $f_0$ be a list $H_0$-colouring of $G$ with respect to $L_0$.  Let $f$ be a
list $(H\setminus F)$-colouring of $G\setminuss F$ with respect to $L$ such that
for each $u\in V(G)$, we have $f(u)\in A$ if $f_0(u)\in A_0$, and $f(u)\in B$ if
$f_0(u)\in B_0$.  Since the lists $L$ are separator-consistent on $F$, such
a colouring can be found.

We make some observations about $f$. First, note that $f(u)=x$ if $f_0(u)=x$
where $x\in X$. Indeed, this follows since $G$ is separable on $F$, and $f_0$
respects the lists $L_0$. Similarly, $f(u)=y$ if $f_0(u)=y\in Y$.  Moreover,
$f(u)\not\in X$ if $f_0(u)=a$, and $f(u)\not\in Y$ if $f_0(u)=b$.  It now
follows that $f$ is a homomorphism of $G$ to $H$ with respect to $L$.

Finally, we observe that the lists $L_0$ are all of size at most two. Hence, the
mapping $f_0$ can be found in polynomial time by the standard reduction to
$2SAT$.  That concludes the proof.
\end{proof}

\section{Dichotomy for trigraph trees}

In this section, we prove the dichotomy for L-HOM$(H)$ for trigraph trees~$H$,
i.e., for trigraphs $H$ whose underlying graph is a tree. Partial results along these 
lines are included in \cite{geordie}; specifically, the case when the underlying graph 
of $H$ is a path is solved there.

Let $H^-$ be the digraph obtained from a trigraph $H$ by removing all edges $xy$
such that at least one of $xy, yx$ is a strong edge of $H$, and by removing all
vertices $x$ such that $xx$ is a strong loop in $H$.

\begin{theorem}\label{thm:trees}
If $H$ is a trigraph tree, then L-HOM$(H)$ is polynomially equivalent to
L-HOM$(H^-)$.
\end{theorem}

\begin{corollary}
If $H$ is a trigraph tree, then L-HOM$(H)$ is polynomial time solvable or
$NP$-complete.
\end{corollary}

\begin{proof}
Suppose that L-HOM$(H)$ is polynomial time solvable. Observe that, since the
underlying graph of $H$ is a tree, each connected component of $H^-$ is an 
induced subgraph of $H$. It follows that L-HOM$(H^-)$ is also polynomial time 
solvable.

Suppose that L-HOM$(H^-)$ is polynomial time solvable.  We prove the theorem by
induction on the size of $V(H)$.  Hence, we shall assume that for each vertex
$x$ of $H$, L-HOM($H-x$) is polynomial time solvable.

Let $G$ with lists $L$ be an instance of L-HOM$(H)$.  If $H$ contains no strong
loops or strong edges, then $H=H^-$ and there is nothing to prove.

Suppose that $H$ contains a strong loop at $x$.  By Lemmata \ref{lem:1},
\ref{lem:2}, and \ref{lem:3} we may assume that the lists $L$ are arc-consistent
and contain representatives and strong representatives.

Consider the strong representative $S_x$ of $x$. Let $B$ denote the vertices
$u\in V(G)$ with $x\in L(u)$. If $x$ has no symmetric neighbours, $S_x=\{x\}$.
Hence, $L(v)=\{x\}$ whenever $x\in L(v)$. Therefore, since the lists $L$ are
arc-consistent, $G$ admits a list $H$-colouring with respect to $L$ if and only
if $G-B$ admits a list $(H-x)$-colouring with respect to $L$.  Since
L-HOM$(H-x)$ is polynomial time solvable, the claim follows.

Now, suppose that $x$ has symmetric neighbours. Since the underlying graph
of $H$ is a tree, it follows that $S_x=\{x,y\}$ where $y$ is a symmetric
neighbour~of~$x$.

First, suppose that $y$ has no loop. We apply Lemma \ref{lem:8} for $X=\{x\}$,
$Y=\{y\}$, and $Z=\emptyset$, and we see that we may assume that the lists $L$
are arc-consistent and sparse-dense-consistent on $\{x\}$ and $\{y\}$. That is,
we have $L(u)=\{x\}$ or $L(u)=\{y\}$ or $L(u)\setminus\{x,y\}\neq\emptyset$ for
each $u\in V(G)$. In particular, because $S_x=\{x,y\}$, we have $L(u)=\{x\}$
whenever $x\in L(u)$.  Since the lists $L$ are arc-consistent, the claim follows
exactly as above.

Next, suppose that $y$ has a weak loop. Let $K$ be the component of $H-x$ to
which $y$ belongs. Observe that $y$ weakly dominates $x$ on $V(K)$. We apply
Lemma \ref{lem:8} for $X=\{x\}$, $Y=\emptyset$, and $Z=\{y\}$. This implies that
we may assume that the lists $L$ are arc-consistent and separator-consistent on
$X$. Moreover, $G-B$ admits a list $(H-x)$-colouring with respect to $L$, since
otherwise LHOM$(H)$ has no solution for $G$ with lists~$L$.

Hence, we let $f_0$ be a list $(H-x)$-colouring of $G-B$ with respect to $L$
such that $f_0$ restricted to any component $C$ of $G-B$ is a $K$-colouring as long
as a vertex of $K$ appears on the list of some vertex of $C$. Since the lists
$L$ are separator-consistent on $X$, such a mapping must exist. In fact, $f_0$ can
be constructed in polynomial time, since L-HOM$(H-x)$ is polynomial time
solvable.

We extend this mapping to $G$ and show that this yields a homomorphism. 
Let $f$ be a mapping defined as follows.

$$f(u)=\left\{\begin{array}{l@{\quad}l}
	x & {\rm if~}u\in B~{\rm and}~L(u)=\{x\}\\
	y & {\rm if~}u\in B~{\rm and}~L(u)=\{x,y\}\\
	f_0(u) & {\rm otherwise}
	\end{array}\right.
$$

Note that since $S_x=\{x,y\}$, we have for $u\in B$ either $L(u)=\{x\}$ or
$L(u)=\{x,y\}$. Hence, the mapping $f$ is well-defined.  Moreover, $f$ clearly
respects the lists $L$. We show that $f$ is also a homomorphism.

Suppose that $f(u)f(v)\not\in W(H)\cup S(H)$ for some $uv\in E(G)$.  Clearly, at
least one of $u,v$ must belong to $B$, since $f_0$ is a homomorphism.  Also, if
$L(u)=\{x\}$, or $L(v)=\{x\}$, or both $u,v\in B$, we obtain a contradiction by
arc-consistency of $L$.  Hence, we can assume, by symmetry, that $u\in B$,
$v\not\in B$, and $f(u)=y$.  Since the list $L$ are arc-consistent, there exists
$s\in L(v)$ such that $ys\in W(H)\cup S(H)$.  In particular, $s\neq x$ since
$v\not\in B$, and hence, $s$ belongs to $K$. Therefore, if $C$ is the component
of $G-B$ to which $v$ belongs, $f_0$ restricted to $C$ is a $K$-colouring. In
particular, $f(v)=f_0(v)\in V(K)$.

Now, if $xf(v)\in S(H)$, both $x$ and $y$ cannot belong to $L(u)$ by
arc-consistency. If $xf(v)\in W(H)$, then $yf(v)\in W(H)$ since $y$ weakly
dominates $x$ on $V(K)$. Therefore, $xf(v)\not\in W(H)\cup S(H)$. But, since the
lists $L$ are arc-consistent, and $L(u)=\{x,y\}$, we obtain a contradiction.

Suppose now that $f(u)f(v)\in S(H)$ for $uv\not\in E(G)$. Clearly, $u\in B$ or
$v\in B$. Again, if $L(u)=\{x\}$, or $L(v)=\{x\}$, or $u,v\in B$, we have a
contradiction. Hence, we can assume, by symmetry, that $u\in B$, $v\not\in B$,
and $f(u)=y$. Also, since $yf(v)$ is an edge and $v\not\in B$, $f(v)$ belongs to
$K$. As above, $xf(v)$ can neither be a non-edge nor a weak edge
by arc-consistency of $L$, respectively the fact that $y$ weakly dominates $x$
on $V(K)$. Therefore, $xf(v)\in S(H)$, and again, since the lists $L$ are
arc-consistent and $L(u)=\{x,y\}$, we have a contradiction.

This proves that $f$ is indeed a list $H$-colouring of $G$ with respect to $L$,
and clearly, $f$ can be constructed in polynomial time.

Next, suppose that $y$ has a strong loop. Consider the strong representative
$S_y$. If~$S_y=\{y\}$ or $S_y=\{y,x'\}$ where $x'\neq x$, it follows that
$L(u)=\{x\}$ whenever $x\in L(u)$. In this case, the claim follows, as above,
from arc-consistency of $L$. Hence, we may assume $S_x=S_y=\{x,y\}$. We observe
that $H\setminus \{xy,yx\}$ contains two components, a component $A$ which
contains $x$, and a component $B$ which contains $y$.  Therefore, if we let
$X=\{x\}$ and $Y=\{y\}$, then the list $L(v)$ of any vertex $v\in V(G)$ either
does not contain both $x$ and $y$, or $L(v)\subseteq\{x,y\}$. This shows that
$G$ is separable on $F=\{xy\}$.  In addition, L-HOM$(H-\{x,y\})$ is polynomial
time solvable. Therefore, by Lemma \ref{lem:9}, we can solve L-HOM$(H)$ for $G$
with the lists $L$ in polynomial~time.\medskip

Finally, suppose that $H$ contains a strong edge $xy$, and $yx$ is an edge of
$H$. Let $A$ be the component of $H\setminus \{xy,yx\}$ which contains $x$, and $B$ be
the component which contains $y$.  Let $X=\{x\}$ and $Y=\{y\}$. If $x$ does not
appear on the list of any vertex in $G$, then $G$ with the lists $L$ is an instance
of \mbox{L-HOM$(H-x)$} which is polynomial time solvable.  Similarly for $y$.
Hence, since the lists $L$ contain representatives, there exist vertices
$v_x,v_y\in V(G)$ such that $L(v_x)=\{x\}$ and $L(v_y)=\{y\}$.  Now, suppose
that the list $L(u)$ contains~$x$.  By arc-consistency, we must have $uv_y\in
E(G)$.  Since $y$ is not adjacent to any vertex of $A\setminus X$ and $uv_y\in
E(G)$, no vertex of $A\setminus X$ can appear on the list $L(u)$.  Similarly, if
$y\in L(u)$, no vertex of $B\setminus Y$ appears on the list $L(u)$.  This shows
that $G$ is separable on $F=\{xy,yx\}$. If $yx$ is not an edge of $H$, the same
argument shows that $G$ is separable on $F=\{xy\}$. In both cases, the claim
follows by Lemma~\ref{lem:9}.

That concludes the proof.
\end{proof}

\section{Extensions}

In this section, we extend the dichotomy from the previous section to a larger
class of trigraphs which we refer to as {\em tree-like}.  In fact, we prove the
dichotomy for this class using the same tools we used in the previous section.
We have stated these tools in a sufficiently general way to easily allow for
this extension.\bigskip

\noindent{\em Domination property}\medskip

We say that an ordering $x_1,\ldots,x_t$ of vertices is a {\em domination
ordering}, if $x_i$ weakly dominates $x_j$ whenever $i<j$.

Let $H$ be a trigraph, and suppose that $x$ is strong loop of $H$.
For a component $K$ of $H-x$, let $R_K$ denote the set of all 
symmetric neighbours of $x$ in $K$. Recall that according to our 
definitions, $y$ is a symmetric neighbour of $x$ just if both $xy$ 
and $yx$ are in $W(H)\cup S(H)$.

We say that a component $K$ satisfies {\em the domination property} if 

\begin{enumerate}[(D1)]
\item no vertex of $R_K$ has a strong loop, and
\item the vertices of $R_K$ with weak loops admit a domination ordering, each
weakly dominates $x$ on $V(K)$, and each weakly dominates every vertex of 
$R_K$ without a loop.
\end{enumerate}\medskip

\noindent{\em Matching property}\medskip

Let $F$ be the edges of $H$. We say that $F$ {\em separates $H$} if each edge
$xy\in F$ has its endpoints $x,y$ in different components of $H\setminus F$.
Thus $F$~separates $H$ if and only if there is a partition of $V(H)$ such that
$F$ is the set of all edges between different parts.

Let $F^*$ consist of all edges $xy \in F$ such that neither $xy$ nor $yx$
is a strong edge of $H$. 
We say that $F$ satisfies {\em the matching property} if
\begin{enumerate}[(M1)]
\item $F$ separates $H$,
\item for any $xy\in F^*$, both $xx$ and $yy$ are strong loops,
\item if $xy, zy\in F$, and neither $xy$ nor $zy$ is a bridge of $H$, then $xz$ is not 
a symmetric edge, and $xx$ or $zz$ is a strong loop,
\item if $xy, xz\in F$, and neither $xy$ nor $xz$ is a bridge of $H$, then $yz$ is not 
a symmetric edge, and $yy$ or $zz$ is a strong loop.
\item if $xy\in F^*$, and $xz$ and $yw$ are symmetric edges with $xz \not\in
F$, then $zw\not\in F^*$ and $xw\not\in F^*$.
\item if $xy\in F^*$, and $xz$ and $yw$ are symmetric edges with $yw\not\in F$,
then $zw\not\in F^*$ and $zy\not\in F^*$.
\end{enumerate}\medskip

\subsection{Tree-like trigraphs}~\vspace{-1.5ex}

We are now in a position to define the class ${\mathcal T}$ of {\em tree-like
trigraphs}.

\begin{enumerate}[(T1)]
\item If $H$ has no strong edges (or loops), then $H \in {\mathcal T}$

\item If $H$ contains a strong loop at $x$ such that $H-x \in {\mathcal
T}$ and each component of $H-x$ satisfies the domination property,
then $H \in {\mathcal T}$.

\item If $H$ contains a set $F$ of edges such that $H \setminus F \in {\mathcal
T}$ and the set $F$ satisfies the matching property, then $H \in {\mathcal T}$.
\end{enumerate}

\begin{theorem}\label{thm:4}
If $H$ is a tree-like trigraph, i.e., if $H\in\mathcal T$, then \mbox{L-HOM$(H)$}
is polynomially equivalent to L-HOM$(H^-)$.  In particular, \mbox{L-HOM$(H)$} is
polynomial time solvable or $NP$-complete.
\end{theorem}

\begin{proof}
We prove the theorem by structural induction on $H$.  We shall assume that $H$
is connected, otherwise we treat each component of $H$ separately.  If $H$ is a
digraph, there is nothing to prove. 

Suppose that $H$ contains a strong loop at $x$ such that $H-x\in\mathcal T$ and
each component of $H-x$ satisfies the domination property. If L-HOM$(H)$ is
polynomial time solvable, then so is L-HOM$(H-x)$. On the other hand, since
$H-x\in\mathcal T$, we shall assume, by induction, that L-HOM$(H-x)$ is
polynomial time solvable.

Let $G$ with lists $L$ be an instance of L-HOM$(H)$.  By Lemmata \ref{lem:1},
\ref{lem:2}, and \ref{lem:3}, we may assume that the lists $L$ contain
representatives and strong representatives.

Consider the strong representative $S_x$.  Let $B$ be the vertices $u\in V(G)$
with $x\in L(u)$.  It follows that there is a unique connected component $K$ of
$H-x$ such that $S_x\setminus\{x\}\subseteq K$.  Let $X=\{x\}$, let $Y$ be the
vertices of $S_x$ with no loops, and $Z$ be the vertices of $S_x$ with weak
loops. Since $K$ satisfies the domination property, we have $S_x=X\cup Y\cup Z$.
Also, $Z$ admits a domination ordering $z_1,\ldots,z_{|Z|}$, and each vertex of
$Z$ dominates every vertex of $Y$. For $u\in V(G)$, let $\min(u)$ denote the
vertex $z_i\in L(u)$ with the smallest index $i$. Now, by Lemma \ref{lem:6}, we
may assume that for each $u\in V(G)$, either $L(u)=\{x\}$, or
$L(u)=\{x,\min(u)\}$, or $x\not\in L(u)$.

Moreover, L-HOM$(H[X])$ and L-HOM$(H-X)$ are polynomial time solvable since
$X=\{x\}$, and by the inductive hypothesis, respectively.  Hence, by applying Lemma 
\ref{lem:8} to $X,Y,Z$, we see that may assume that the lists $L$ are arc-consistent, 
separator-consistent on $X$, and sparse-dense-consistent on $X$~and~$Y$.

Now, let $f_0$ be a list $(H-x)$-colouring of $G-B$ with respect to $L$, such
that $f_0$ reduced to any component $C$ of $G-B$ is a $K$-colouring as long as a
vertex of $K$ appears on the list of some vertex of $C$. Since the lists $L$ are
separator-consistent on $X$, such mapping must exist and can be found in
polynomial time.

We extend the mapping $f_0$ to a mapping $f$ as follows.
$$f(u)=\left\{\begin{array}{l@{\quad}l}
	x & {\rm if~}u\in B{\rm ~and~}L(u)=\{x\}\\
	\min(u) & {\rm if~}u\in B{\rm~and~}L(u)=\{x,\min(u)\}\\
	f_0(u) & {\rm otherwise}
	\end{array}\right.
$$
By the above remark, the mapping $f$ is well-defined. It also clearly respects
the lists $L$.

It remains to show that $f$ is a homomorphism.  The proof of this follows
exactly as in the proof of Theorem \ref{thm:trees}, and hence, we skip
further~details.\medskip 

Now, suppose that $H$ contains edges $F$ satisfying the matching property and
such that $H\setminus F\in\mathcal T$. Again, we assume that L-HOM$(H\setminus
F)$ is polynomial time solvable, and that the lists $L$ are arc-consistent,
contain representatives and strong representatives.

First, we observe that any subset $F_0$ of $F$ satisfies the matching property
in $H'=H\setminus (F\setminus F_0)$. In particular, $H'\setminus F_0=H\setminus
F\in \mathcal T$.

Therefore, we can proceed by induction on $F$.  We denote by $V_F$ the vertices
of $H$ with at least one incident edge from $F$.

We observe that if a vertex $x\in V_F$ does not appear on any list, we can
remove from $F$ all edges incident to $x$. Since there is at least one such
edge, the claim follows by induction.  Hence, we shall assume that each vertex
of $V_F$ appears on some list.

Let $x\in V_F$ be a strong loop. Consider the strong representative $S_x$.  We
claim that either $S_x=\{x,y\}$ where $xy\in F$, or $S_x$ belongs to a component
of $H\setminus F$.  To prove this, suppose that $y,y'\in S_x\setminus\{x\}$ and
$xy\in F$.  Clearly, all of $xy$, $xy'$, and $yy'$ are symmetric edges. Since $F$
separates $H$, at least one of $xy',yy'$ must be in $F$. But that contradicts
the matching property.

Now, suppose that $S_x$ belongs to a component of $H\setminus F$. Let $Q_x$
denote the union of lists of vertices $u\in V(G)$ with $x\in L(u)$. We have
$Q_x\subseteq S_x$. Since $x\in V_F$, there exists $y$ with $xy\in F$ or $yx\in F$. 

By symmetry, suppose that $xy\in F$, and suppose also that $xy$ is a weak edge.
Hence, $y$ is a strong loop.  Consider the strong representative~$S_y$.  Suppose
first that $S_y$ also belongs to a component of $H\setminus F$.  Clearly,
$S_x\cap S_y=\emptyset$.  In particular, $Q_x\cap Q_y=\emptyset$ where~$Q_y$ is
the union of lists $L(u)$ with $y\in L(u)$.  Furthermore, for $z\in Q_x\setminus
\{x\}$ and $w\in Q_y\setminus\{y\}$, we have by the matching property that
$zy,xw$ are not edges of $H$, and either $zw\in S(H)$ or $zw$ is not an edge of
$H$. In particular, if $zw\in S(H)$, then arc-consistency of $L$ implies
$z\not\in Q_x$.  Hence, it follows that the only edge in $H$ from $Q_x$ to $Q_y$
is the edge $xy$. This implies that for any $uv\in E(G)$ with $x\in L(u)$ and
$y\in L(v)$, we have $L(u)=\{x\}$ and $L(v)=\{y\}$ by arc-consistency of $L$.
Therefore, we can remove all such edges $uv$, and after that, we can remove $xy$
from $F$. The claim now follows~by~induction.

Next, suppose that $Q_y=\{y,w\}$ where $wy\in F$. Let $z\in Q_x\setminus\{x\}$.
By the matching property, we have $zy\not\in W(H)\cup S(H)$. Also, either $zw\in
S(H)$ or $zw$ is not an edge, and either $xw\in S(H)$ or $xw$ is not an edge. If
$zw\in S(H)$, then arc-consistency of $L$ implies $w\not\in Q_y$. If $xw\in
S(H)$ and $zw$ is not an edge, then arc-consistency of $L$ implies $z\not\in
Q_x$. Again, the only edge from $Q_x$ to $Q_y$ is the edge $xy$, and the claim
follows by induction.\medskip

Therefore, we may assume that any strong loop $x\in V_F$ is either incident to a
strong edge of $F$, or we have $S_x=\{x,y\}$ where $xy\in F$.

Recall that $F$ separates $H$.  First, suppose that $H\setminus F$ contains
exactly two components. We prove that $G$ is separable on $F$. Suppose that we
have $x\in L(u)$ where $x\in V_F$. Let $K$ be the component of $H\setminus F$
which contains $x$. Suppose that there exists $y$ such that $xy\in F$ or $yx\in
F$ and one of $xy,yx$ is strong. Then, by the matching property and
arc-consistency of $L$, the list $L(u)$ contains no vertex of $K$ other than
$x$. If no such $y$ exists, then $xx$ must be a strong loop. Hence, by the above
assumption, we have $S_x=\{x,y\}$ where $xy\in F$.  Therefore, $L(u)$ again
contains no vertex of $K$ other than $x$.  The claim now follows by Lemma
\ref{lem:9}.

Finally, suppose that $H\setminus F$ contains more than two components.  Let
$F_0$ be a smallest subset of $F$ such that $H\setminus F_0$ is disconnected.
Since $H$ is connected and $H\setminus F$ is disconnected, the set $F_0$ must
exist.  As remarked above, $F\setminus F_0$ satisfies the matching property in
$H\setminus F_0$.  Also, it can be seen that $F_0$ satisfies the matching
property in $H$.  Since $F_0$ contains at least one edge and $H\setminus F_0$
contains exactly two components because of minimality, the claim follows
by~induction.

That concludes the proof.
\end{proof}

The class ${\mathcal T}$ admits a few natural extensions which we
shall only mention tangentially. For instance, we shall observe the
following simple fact.

\begin{theorem}\label{thm:simple}
If each vertex of $H$ has a strong loop, and the symmetric graph of $H$ contains
no triangles, then L-HOM$(H)$ is polynomial time solvable.
\end{theorem}
\begin{proof}
Let $G$ with lists $L$ be an instance of L-HOM$(H)$.  By Lemma \ref{lem:3}, we
may assume that the lists $L$ contain strong representatives. Since the
symmetric graph of $H$ has no triangles, it follows that for each $x\in V(H)$,
the strong representative $S_x$ contains at most two elements. The problem now
can be reduced to $2SAT$ which is polynomial time solvable.
\end{proof}

As a consequence, we can extend the class ${\mathcal T}$ by adding another basis
clause to its recursive description:

\begin{enumerate}[(T1$^\prime$)]
\item
If each vertex of $H$ has a strong loop, and the symmetric graph of $H$ contains
no triangles, then $H \in {\mathcal T}$.
\end{enumerate}

Similarly, if $H$ has no weak edges, i.e., if $W(H)=\emptyset$, then L-HOM$(H)$
is polynomially solvable \cite{fulcsp,fhkm}.  Hence, we can also add the
following clause:

\begin{enumerate}[(T2$^\prime$)]
\item
If $H$ has no weak edges, then $H \in {\mathcal T}$.
\end{enumerate}\medskip

\subsection{Trigraph trees and special tree-like trigraphs}~\vspace{-1.5ex}

Since the recursive description of the class ${\mathcal T}$ is complex,
we shall identify a subclass of ${\mathcal T}$ which can be defined directly.

Let $F(H)$ denote all edges $xy$ of $H$ such that either 
\begin{enumerate}[(i)]
\item
$xy$ or $yx$ is a strong edge of $H$, or 
\item
$xy$ and $yx$ are weak edges of $H$, and $x$ and $y$ are strong loops of $H$.
\end{enumerate}

We say that $H$ is a {\em special tree-like trigraph} if
there is a set $F'\supseteq F(H)$ of edges of $H$ such that
\begin{enumerate}[(S1)]
\item $F'$ satisfies the matching property, and
\item for every strong loop $x$ of $H\setminus F'$, each component $K$ of
$H\setminus F'-x$ satisfies the domination property in $H\setminus F'$.
\end{enumerate}

Let $\mathcal S$ denote the class of all special tree-like trigraphs. Also, let
$\mathcal T_0$ denote the class of all trigraph trees.

\begin{theorem}
$\mathcal T_0\subseteq \mathcal S\subseteq \mathcal T$.
\end{theorem}
\begin{proof}
First, let $H$ be in $\mathcal T_0$. Consider the set $F'=F(H)$.  Since each
edge in $F'$ is a bridge of $H$, conditions (M1), (M3) and (M4) are satisfied for $F'$.
Also, the edges of $F'$ do not form cycles in $H$ since $H$ is a trigraph tree,
and hence, (M5) and (M6) are satisfied. Therefore, $F'$ satisfies the matching
property. On the other hand, for every vertex $x$ of $H$ with a strong loop,
each component of $H-x$ satisfies the domination property, since $x$ is adjacent
to at most one vertex of this component. It follows that~$H\in\mathcal S$.

Now, for $H\in\mathcal S$, it suffices to observe that edges $xy$ of $F'\setminus
F(H)$ have both $x$ and $y$ strong loops. Hence, after removing $F'$ using (T3)
and then removing all vertices with strong loops using (T2), we obtain precisely
$H^-$ which is a digraph. Hence, by (T1), we conclude that $H\in \mathcal T$.
\end{proof}

\begin{corollary}
If $H$ is a special tree-like trigraph, i.e., if $H\in\mathcal S$, then
\mbox{L-HOM$(H)$} is polynomial time solvable or $NP$-complete.
\end{corollary}

Although some simple extensions of tree-like, or special tree-like trigraphs are
possible (such as, say, Theorem \ref{thm:simple}), we shall mention in the
conclusions some example trigraphs outside ${\mathcal T}$ for which the first
part of Theorem \ref{thm:4} fails.\medskip

\subsection{Representatives of strong edges}~\vspace{-1.5ex}

In this section, we describe an extension of the notion of strong
representatives to strong edges of trigraphs.

The {\em underlying trigraph $H'$ of $H$} is the trigraph on the vertices of $H$
with strong edges $xy$ such that $xy\in S(H)$ or $yx\in S(H)$, and with weak
edges $xy$ such that $xy,yx\not\in S(H)$ and $xy\in W(H)$ or $yx\in W(H)$.

We denote by $G^2$ the digraph on the vertices of $G$ with edges $xy$ such that
$xy\in E(G)$ or $xz,zy\in E(G)$ for some $z\in V(G)$.

We denote by $H^2$ the trigraph on the vertices of $H$ with strong edges $xy$
such that $xy\in S(H)$ or $xz,zy\in S(H)$ for some $z\in V(G)$, and with weak
edges $xy$ such that $xy\not\in S(H^2)$, and $xy\in W(H)$ or $xz,zy\in W(H)\cup
S(H)$ for some $z\in V(G)$.

We say that an edge $xy$ of $H$ is {\em admissible}, if there exist vertices
$u,v\in V(G)$ such that $x\in L(u)$ and $y\in L(v)$.

We say that lists $L$ contain {\em representatives for strong edges}, if for
each admissible strong edge $xy$ in $H$, there is a set $S_{xy}$ such that

\begin{enumerate}[(i)]
\item $L(v)\subseteq S_{xy}$ whenever $x\in L(v)$ or $y\in L(v)$,
\item each vertex of $S_{xy}\setminus \{x,y\}$ is a neighbour of $x$ or a
neighbour~of~$y$,
\item if $x$ and $y$ have no common neighbours, then $S_{xy}\backslash\{x,y\}$
contains only neighbours of $x$ or only neighbours of $y$,
\item if $yx$ is a strong edge, then each vertex of $S_{xy}\backslash \{x,y\}$ 
is a symmetric neighbour of $x$ or a symmetric neighbour of $y$,
\item if $yx$ is a strong edge, and $x$ and $y$ have no common symmetric neighbours, then 
\mbox{$S_{xy}\backslash\{x,y\}$} contains only symmetric neighbours of $x$ or only
symmetric \mbox{neighbours}~of~$y$.
\end{enumerate}

\begin{lemma}\label{lem:10}
Lists $L$ can be transformed to contain representatives for strong edges.
\end{lemma}
\begin{proof}
By Lemmata \ref{lem:1} and \ref{lem:2}, we may assume that the lists $L$ are
arc-consistent and contain representatives.

Let $f$ be a list $H$-colouring of $G$ with respect to $L$, and let $xy$ be an
admissible strong edge of $H$. Let $C=f^{-1}(x)\cup f^{-1}(y)$, let $N$ denote
all vertices of $V(H)$ that are neighbours of $x$ or $y$, and let $B$ denote all
vertices $u\in V(G)$ with $x\in L(u)$ or $y\in L(u)$.  We have $C\subseteq B$.
Also, since the lists $L$ are arc-consistent, we have $L(u)\subseteq N$ for each
$u\in B$.

Let $G'$ be the underlying graph of $G$ and $H'$ be the underlying trigraph of
$H$. It is easy to verify that $f$ is a homomorphism of $G'$ to $H'$.  In
particular, $f$ is a surjective mapping from of $B$ to $N_0$, where $N_0=f(B)$.
Clearly, $N_0\subseteq N$.

Now, it is not difficult to prove that $f$ is a homomorphism of $(G'[B])^2$ to
$(H'[N_0])^2$. (For this, we need the above remark about surjectivity.) We
observe that $C$ induces a clique in $(G'[B])^2$. Hence, using the same argument
as in the proof of Lemma \ref{lem:3}, there is a maximal clique $C''$ of a
minimal chordal completion $G''$ of $(G'[B])^2$ such that for each $u\in
C''\setminus C$, we have $f(u)\in V(K)$ where $K$ is a unique component of
$(H'[N_0])^2-\{x,y\}$. In particular, if $x$ and $y$ have no common neighbours,
then $V(K)$ either contains only neighbours of $x$ or it contains only
neighbours of $y$. Hence, we let $S_{xy}=V(K)\cup\{x,y\}$.

Now, if $yx$ is also a strong edge, we obtain $S_{xy}$ by replacing $G$
with the symmetric graph of $G$, and replacing $H$ with the symmetric trigraph
of $H$. The remainder of the proof follows exactly as in Lemma~\ref{lem:3}.
\end{proof}

\subsection{Trigraph cycles}~\vspace{-1.5ex}

A trigraph $H$ is a {\em trigraph cycle} if the underlying graph of $H$ is a
cycle.  Let $H$ be a trigraph cycle. We say that $H$ is a {\em good cycle} if
$H$ has at least one of the following:

\begin{enumerate}[(i)]
\item two strong edges, or
\item three consecutive strong loops, or
\item two pairs of consecutive strong loops, or
\item a strong edge and a distinct pair of consecutive strong loops, or
\item two strong loops joined by a nonsymmetric edge, or
\item a strong loop whose neighbours have no loops, or
\item a strong loop with non-symmetric edges to neighbours, or
\item a strong edge with at least one endpoint having no loop.
\end{enumerate}

\begin{theorem}
If $H$ is a good cycle, then the problem L-HOM$(H)$ is polynomial time solvable
or $NP$-complete.
\end{theorem}

Figure \ref{fig:cycle1} illustrates example trigraph cycles whose complexity is
not determined by our theorem. These are cycles $H$ that contain vertices $x,y$
such that either $xy\in S(H)$ and $xx,yy\in W(H)$, or all of $xx$, $xy$, $yx$,
and $yy$ are edges (weak or strong) but at least one is strong. (Only three
typical cases of this are shown in the figure.)

We prove the theorem by reducing the problem to an induced subgraph of $H$.
Unfortunately, if $H$ is a trigraph cycle that is not a good cycle, such
reduction may not be possible at all. The~complement of the stable cutset
problem (Figure \ref{fig:ex1}c) is a good example illustrating
this~difficulty.\medskip

\begin{proof}
We assume that $H$ has at least five vertices, since otherwise the claim follows
from \cite{cehm}.  Let $G$ with lists $L$ be an instance of L-HOM$(H)$. By
Lemmata \ref{lem:1}, \ref{lem:2}, \ref{lem:3}, and \ref{lem:10}, we assume that
the lists $L$ are arc-consistent, contain representatives, strong
representatives, and representatives for strong edges.  Also, we assume that
each vertex $x$ of $H$ appears on some list, since otherwise we can reduce the
problem $H-x$, which is a trigraph tree, and the claim follows by Theorem
\ref{thm:trees}.

Suppose that $H$ contains two strong edges $e$, $e'$. If $F=\{e,e'\}$ satisfies
the matching property, then we are done by Theorem \ref{thm:4}. Otherwise,
we must have $e=xy$ and $e'=zy$. (The case $e=yx$ and $e'=yz$ is symmetric.) Let
$t$ be any vertex (including $x,z$) of $H$ other than $y$.  By arc-consistency
of $L$, no list $L(u)$ contains both $y$ and $t$ since $xy$ and $zy$ are strong,
but at least one of $xt,zt$ is not an edge.  Hence, we have $L(u)=\{y\}$
whenever $y\in L(u)$. By arc-consistency of $L$, we can reduce the problem to
$H-y$ and we are done.

Next, if $H$ contains two pairs of consecutive strong loops $x,y$ and $z,w$ with
possibly $y=z$, or a strong edge $xy$ and a pair of strong loops $z,w$ where
$\{z,w\}\neq\{x,y\}$, then $F=\{xy,zw\}$ satisfies the matching property, and
again we are done.

Hence, suppose that $H$ contains strong loops $x,y$ where $xy$ is an edge, but
$yx$ is not. Consider the strong representatives $S_x$ and $S_y$. Clearly, $S_x$
and $S_y$ are disjoint. Also, the only edge from $S_x$ to $S_y$ is the edge
$xy$. Hence, by arc-consistency of $L$, if $uv\in E(G)$ with $x\in L(u)$ and
$y\in L(v)$, we have $L(u)=\{x\}$ and $L(v)=\{y\}$. In particular, we can remove
all such edges $uv$, and afterwards, we remove the edge $xy$ from $H$. Now,
since $H\setminus xy$ is a trigraph tree, by Theorem \ref{thm:trees}, we
conclude that L-HOM$(H\setminus xy)$ is polynomially equivalent to L-HOM$(H-x)$.
Hence, we are done.

Next, suppose that $H$ contains a strong loop $x$ with neighbours $y,y'$ having
no loops. Consider the strong representative $S_x$. If $S_x=\{x,y\}$, then we
apply Lemma \ref{lem:8} for $X=\{x\}$, $Y=\{y\}$, and $Z=\emptyset$. This yields
that the lists $L$ are arc-consistent and sparse-dense-consistent on $X$ and
$Y$. In particular, $L(u)=\{x\}$ whenver $x\in L(u)$. Similarly, if
$S_x=\{x,y'\}$.  Hence, we can reduce the problem to $H-x$ and we are done.

Now, suppose that $H$ contains a strong loop $x$ with non-symmetric edges
between $x$ and its neighbours $y,y'$. Clearly, we have $S_x=\{x\}$. Hence, we
can reduce the problem to $H-x$ and we are done.

Finally, suppose that $H$ contains a strong edge $xy$ whose one endpoint has no
loop. By symmetry suppose that $x$ has no loop. Let $z$ be the other neighbour
of $x$ and let $w$ be the other neighbour of $y$. Consider the representative
$S_{xy}$.  We have that either $S_{xy}=\{x,y,w\}$ or $S_{xy}=\{z,x,y\}$. By
arc-consistency of $L$, no list $L(u)$ contains both $y$ and $w$ since $xy$ is
strong, but $xw$ is not an edge.  Similarly, no list $L(u)$ contains both $x$
and $y$, or both $x$ and $z$.  Hence, if $S_{xy}=\{z,x,y\}$, we have
$L(u)=\{x\}$ whenever $x\in L(u)$, and if $S_{xy}=\{x,y,w\}$, we have
$L(u)=\{y\}$ whenever $y\in L(u)$. Therefore, we can reduce the problem to $H-x$
or to $H-y$, and we are done.
\end{proof}

\begin{figure}[t!]
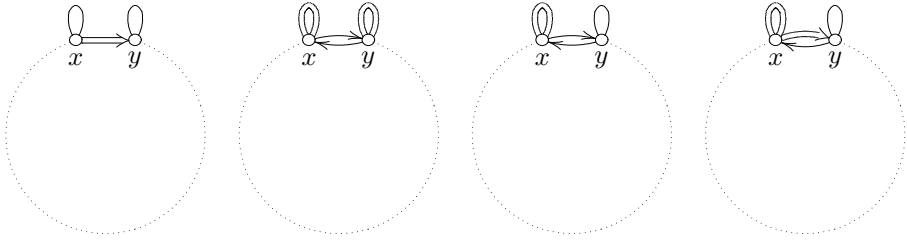

\centering
$
\xy/r3pc/:
(1.5,0);(0,0);{\xypolygon10"a"{~>{}{\phantom{s}}}};
"a3";"a4";{\ellipse :a(-55),^{.}};
{\ar@{=>} "a4";"a3"};
"a3"*[o][F]{\phantom{s}};
"a4"*[o][F]{\phantom{s}};
{\ar@{-}@/^0.1pc/ "a4";"a4"+(-0.07,0.3)};
{\ar@{-}@/^0.2pc/ "a4"+(-0.07,0.3);"a4"+(0.07,0.3)};
{\ar@{-}@/_0.1pc/ "a4";"a4"+(0.07,0.3)};
{\ar@{-}@/^0.1pc/ "a3";"a3"+(-0.07,0.3)};
{\ar@{-}@/^0.2pc/ "a3"+(-0.07,0.3);"a3"+(0.07,0.3)};
{\ar@{-}@/_0.1pc/ "a3";"a3"+(0.07,0.3)};
"a4"+(0,-0.2)*{x};
"a3"+(0,-0.2)*{y};
\endxy
\quad
\xy/r3pc/:
(1.5,0);(0,0);{\xypolygon10"a"{~>{}{\phantom{s}}}};
"a3";"a4";{\ellipse :a(-55),^{.}};
{\ar@/^0.15pc/ "a3";"a4"};
{\ar@/^0.15pc/ "a4";"a3"};
"a3"*[o][F]{\phantom{s}};
"a4"*[o][F]{\phantom{s}};
{\ar@{=}@/^0.1pc/ "a4";"a4"+(-0.07,0.3)};
{\ar@{=}@/^0.2pc/ "a4"+(-0.07,0.3);"a4"+(0.07,0.3)};
{\ar@{=}@/_0.1pc/ "a4";"a4"+(0.07,0.3)};
{\ar@{=}@/^0.1pc/ "a3";"a3"+(-0.07,0.3)};
{\ar@{=}@/^0.2pc/ "a3"+(-0.07,0.3);"a3"+(0.07,0.3)};
{\ar@{=}@/_0.1pc/ "a3";"a3"+(0.07,0.3)};
"a4"+(0,-0.2)*{x};
"a3"+(0,-0.2)*{y};
\endxy
\quad
\xy/r3pc/:
(1.5,0);(0,0);{\xypolygon10"a"{~>{}{\phantom{s}}}};
"a3";"a4";{\ellipse :a(-55),^{.}};
{\ar@/^0.15pc/ "a3";"a4"};
{\ar@/^0.15pc/ "a4";"a3"};
"a3"*[o][F]{\phantom{s}};
"a4"*[o][F]{\phantom{s}};
{\ar@{=}@/^0.1pc/ "a4";"a4"+(-0.07,0.3)};
{\ar@{=}@/^0.2pc/ "a4"+(-0.07,0.3);"a4"+(0.07,0.3)};
{\ar@{=}@/_0.1pc/ "a4";"a4"+(0.07,0.3)};
{\ar@{-}@/^0.1pc/ "a3";"a3"+(-0.07,0.3)};
{\ar@{-}@/^0.2pc/ "a3"+(-0.07,0.3);"a3"+(0.07,0.3)};
{\ar@{-}@/_0.1pc/ "a3";"a3"+(0.07,0.3)};
"a4"+(0,-0.2)*{x};
"a3"+(0,-0.2)*{y};
\endxy
\quad
\xy/r3pc/:
(1.5,0);(0,0);{\xypolygon10"a"{~>{}{\phantom{s}}}};
"a3";"a4";{\ellipse :a(-55),^{.}};
{\ar@{->}@/^0.2pc/ "a3";"a4"};
{\ar@{=>}@/^0.2pc/ "a4";"a3"};
"a3"*[o][F]{\phantom{s}};
"a4"*[o][F]{\phantom{s}};
{\ar@{=}@/^0.1pc/ "a4";"a4"+(-0.07,0.3)};
{\ar@{=}@/^0.2pc/ "a4"+(-0.07,0.3);"a4"+(0.07,0.3)};
{\ar@{=}@/_0.1pc/ "a4";"a4"+(0.07,0.3)};
{\ar@{-}@/^0.1pc/ "a3";"a3"+(-0.07,0.3)};
{\ar@{-}@/^0.2pc/ "a3"+(-0.07,0.3);"a3"+(0.07,0.3)};
{\ar@{-}@/_0.1pc/ "a3";"a3"+(0.07,0.3)};
"a4"+(0,-0.2)*{x};
"a3"+(0,-0.2)*{y};
\endxy
$
\caption{Unresolved trigraph cycles.\label{fig:cycle1}}
\end{figure}

\subsection{Surjective list homomorphism}~\vspace{-1.5ex}

Finally, we describe how we can use our results to classify the complexity of
finding surjective list homomorphisms for some trigraphs.

We say that a homomorphism $f$ of $G$ to $H$ is a {\em a (vertex) surjective
homomorphism} if $f$ is a surjective mapping of $V(G)$ onto $V(H)$.

The {\em surjective list $H$-colouring problem} SL-HOM$(H)$ takes as input 
a digraph $G$ with lists $L$, and asks whether or not $G$ admits a surjective
list homomorphism to $H$ with respect to $L$.

Let $H^{--}$ be the digraph obtained from a trigraph $H$ by removing all
vertices $x$ with a strong loop at $x$ or a strong edge $xy$ or $yx$ for
some $y$.

\begin{theorem}\label{thm:5}
If $H$ is a special tree-like trigraph, i.e., if $H\in\mathcal S$, then
\mbox{SL-HOM$(H)$} is polynomially equivalent to SL-HOM$(H^{--})$. In
particular, SL-HOM$(H)$ is polynomial time solvable or $NP$-complete.
\end{theorem}

\begin{proof}
Let $H$ be any trigraph (not necessarily a special tree-like).  First, we
observe that SL-HOM$(H)$ is polynomially reducible to L-HOM$(H)$, since
SL-HOM$(H)$ is a special case of L-HOM$(H)$.

Now, suppose that $H$ contains a digraph $H_0$ as an induced subgraph.  Let
$G_0$ with lists $L_0$ be an instance of L-HOM$(H_0)$.  Let $H'$ be the digraph
associated with $H$, and let $G$ be the disjoint union of $G_0$ and $H'$. Define
$L(x)=\{x\}$ for $x\in V(H')$, and  $L(x)=L_0(x)$ for $x\in V(G_0)$.  It now
follows that $G_0$ admits a list $H_0$-colouring with respect to $L_0$ if and
only if $G$ admits a surjective list $H$-colouring with respect to $L$.

This yields that L-HOM$(H^{--})$ is polynomially equivalent to SL-HOM$(H^{--})$.
In fact, we can also conclude that SL-HOM$(H^{--})$ is polynomially reducible to
SL-HOM$(H)$.

Now, let $H$ be a special tree-like digraph. Let $F'\supseteq F(H)$ be the set
of edges from the definition of $H$. Let $G$ with lists $L$ be an instance of
SL-HOM$(H)$. We can assume that each vertex of $H$ appears on some list, since
otherwise there is no solution.  Now, following the proof of
Theorem~\ref{thm:4}, we conclude (depending on $L$) that there is a set
$F''\subseteq F'$ such that the instance $G,L$ of L-HOM$(H)$ is polynomially
reducible to an instance of L-HOM$(H\setminus F''$). In fact, the proof implies
the instance is reducible to an instance of L-HOM$(H-U)$ where $U$ are the
vertices incident to the edges of $F''$. Moreover, since each vertex of $H$
appears on some list, the edges of $F'\setminus F''$ are only between strong
loops. Hence, $U$ contains all vertices of $H$ incident to strong edges.
Therefore, $(H-U)^-=H^{--}$.  Also, by the definition of $H$, for each strong
loop $x$ of $H-U$, each connected component of $H-U-x$ satisfies the domination
property. Hence, by Theorem \ref{thm:4}, \mbox{L-HOM}$(H-U)$ is polynomially
equivalent to \mbox{L-HOM}$(H^{--})$. Also, as remarked earlier,
\mbox{SL-HOM}$(H^{--})$ is polynomially equivalent to L-HOM$(H^{--})$.  Hence,
\mbox{SL-HOM}$(H)$ is polynomially reducible to SL-HOM$(H^{--})$.

That concludes the proof.
\end{proof}

The proof of this theorem in particular implies the following.
\begin{corollary}
If $H$ is a digraph, then L-HOM$(H)$ and SL-HOM$(H)$ are polynomially equivalent.
\end{corollary}

In fact, similarly, one can prove a stronger statement.
\begin{proposition}
If for no vertices $x,y,z$ of $H$ we have $xy\in S(H)$ and $xz\not\in W(H)\cup
S(H)$, then L-HOM$(H)$ and SL-HOM$(H)$ are polynomially equivalent.
\end{proposition}

Note that if we have vertices $x,y,z$ as described above, then an instance $G$
in which $x$ occurs on some list allows us to use arc-consistency to make sure
that no vertex has both $y$ and $z$ on its list; while if $x$ appears on no
list, this no longer happens.

\begin{figure}
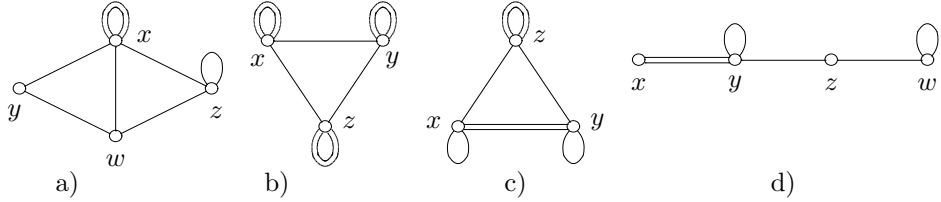

$$
\xy/r3pc/:
(0,1)*[o][F]{\phantom{s}}="x";
(-1,0.5)*[o][F]{\phantom{s}}="y";
(1,0.5)*[o][F]{\phantom{s}}="z";
(0,0)*[o][F]{\phantom{s}}="w";
{\ar@{-} "x";"y"};
{\ar@{-} "x";"z"};
{\ar@{-} "x";"w"};
{\ar@{-} "y";"w"};
{\ar@{-} "z";"w"};
{\ar@{=}@/^0.15pc/ "x";"x"+(-0.1,0.3)};
{\ar@{=}@/^0.25pc/ "x"+(-0.1,0.3);"x"+(0.1,0.3)};
{\ar@{=}@/_0.15pc/ "x";"x"+(0.1,0.3)};
{\ar@{-}@/^0.15pc/ "z";"z"+(-0.1,0.3)};
{\ar@{-}@/^0.25pc/ "z"+(-0.1,0.3);"z"+(0.1,0.3)};
{\ar@{-}@/_0.15pc/ "z";"z"+(0.1,0.3)};
"x"+(0.3,0.05)*{x};
"y"+(-0.05,-0.25)*{y};
"z"+(0.05,-0.25)*{z};
"w"+(0,-0.25)*{w};
(-0.5,-0.5)*{\rm a)};
\endxy
\quad
\xy/r3pc/: (0,0)*{};
(-0.6,1)*[o][F]{\phantom{s}}="x";
(0.6,1)*[o][F]{\phantom{s}}="y";
(0,0.1)*[o][F]{\phantom{s}}="z";
{\ar@{-} "x";"y"};
{\ar@{-} "y";"z"};
{\ar@{-} "x";"z"};
"x"+(-0.1,-0.2)*{x};
"y"+(0.1,-0.2)*{y};
"z"+(0.25,0.05)*{z};
{\ar@{=}@/^0.15pc/ "x";"x"+(-0.1,0.3)};
{\ar@{=}@/^0.25pc/ "x"+(-0.1,0.3);"x"+(0.1,0.3)};
{\ar@{=}@/_0.15pc/ "x";"x"+(0.1,0.3)};
{\ar@{=}@/^0.15pc/ "y";"y"+(-0.1,0.3)};
{\ar@{=}@/^0.25pc/ "y"+(-0.1,0.3);"y"+(0.1,0.3)};
{\ar@{=}@/_0.15pc/ "y";"y"+(0.1,0.3)};
{\ar@{=}@/_0.15pc/ "z";"z"+(-0.1,-0.3)};
{\ar@{=}@/_0.25pc/ "z"+(-0.1,-0.3);"z"+(0.1,-0.3)};
{\ar@{=}@/^0.15pc/ "z";"z"+(0.1,-0.3)};
(-0.5,-0.5)*{\rm b)};
\endxy
~~~
\xy/r3pc/: (0,0)*{};
(-0.6,0.1)*[o][F]{\phantom{s}}="x";
(0.6,0.1)*[o][F]{\phantom{s}}="y";
(0,1)*[o][F]{\phantom{s}}="z";
{\ar@{=} "x";"y"};
{\ar@{-} "y";"z"};
{\ar@{-} "x";"z"};
"x"+(-0.25,0.05)*{x};
"y"+(0.25,0.05)*{y};
"z"+(0.25,0.05)*{z};
{\ar@{-}@/_0.15pc/ "x";"x"+(-0.1,-0.3)};
{\ar@{-}@/_0.25pc/ "x"+(-0.1,-0.3);"x"+(0.1,-0.3)};
{\ar@{-}@/^0.15pc/ "x";"x"+(0.1,-0.3)};
{\ar@{-}@/_0.15pc/ "y";"y"+(-0.1,-0.3)};
{\ar@{-}@/_0.25pc/ "y"+(-0.1,-0.3);"y"+(0.1,-0.3)};
{\ar@{-}@/^0.15pc/ "y";"y"+(0.1,-0.3)};
{\ar@{=}@/^0.15pc/ "z";"z"+(-0.1,0.3)};
{\ar@{=}@/^0.25pc/ "z"+(-0.1,0.3);"z"+(0.1,0.3)};
{\ar@{=}@/_0.15pc/ "z";"z"+(0.1,0.3)};
(0,-0.5)*{\rm c)};
\endxy
\quad
\xy/r3pc/: (0,0)*{};
(-1,0.8)*[o][F]{\phantom{s}}="x";
(0,0.8)*[o][F]{\phantom{s}}="y";
(1,0.8)*[o][F]{\phantom{s}}="z";
(2,0.8)*[o][F]{\phantom{s}}="w";
{\ar@{=} "x";"y"};
{\ar@{-} "y";"z"};
{\ar@{-} "w";"z"};
"x"+(0,-0.25)*{x};
"y"+(0,-0.25)*{y};
"z"+(0,-0.25)*{z};
"w"+(0,-0.25)*{w};
{\ar@{-}@/^0.15pc/ "w";"w"+(-0.1,0.3)};
{\ar@{-}@/^0.25pc/ "w"+(-0.1,0.3);"w"+(0.1,0.3)};
{\ar@{-}@/_0.15pc/ "w";"w"+(0.1,0.3)};
{\ar@{-}@/^0.15pc/ "y";"y"+(-0.1,0.3)};
{\ar@{-}@/^0.25pc/ "y"+(-0.1,0.3);"y"+(0.1,0.3)};
{\ar@{-}@/_0.15pc/ "y";"y"+(0.1,0.3)};
(0.5,-0.5)*{\rm d)};
\endxy
$$
\caption{a) the stubborn problem, b) the complement of the 3-colouring problem,
c)~the~complement of the stable cutset problem, d) a trigraph $H$ with
$NP$-complete L-HOM$(H)$ but polynomial time solvable SL-HOM$(H)$.
\label{fig:ex1}}
\end{figure}

\section{Conclusions}

We investigated the list (and surjective list) homomorphism problems for
trigraphs~$H$. When $H$ is a digraph, we know that each such problem is
polynomial time solvable or $NP$-complete \cite{bulatov}. However, there are
small trigraphs $H$ for which such dichotomy is not known, for instance, the
trigraph in Figure \ref{fig:ex1}a (corresponding to the so-called "stubborn"
problem from \cite{cehm}).

Hence, we have tried to identify properties of trigraphs $H$ which allow us to
prove dichotomy. With this in mind, we have defined the class $\mathcal T$ of
tree-like trigraphs, and we proved that these trigraphs enjoy dichotomy.
Slightly easier to define is the class $\mathcal S$ of special tree-like
trigraphs, included in $\mathcal T$. In particular, the class of trigraphs whose
underlying graph is a tree, is included in $\mathcal S$ (and hence in $\mathcal
T$).

We now offer some tangential evidence that our class $\mathcal T$ carves out a
reasonable portion of trigraphs where L-HOM($H$) is polynomially equivalent to
L-HOM($H^-$).  For instance, one way to violate the matching property is by
having a vertex with a strong loop adjacent to two other vertices with strong
loops joined by a symmetric edge. The trigraph in Figure \ref{fig:ex1}b
illustrates this possibility, and also illustrates that a polynomial time
solvable (in this case trivial) problem L-HOM($H^-$) can arise from an
$NP$-complete problem L-HOM($H$). (In this case, $G$ admits a homomorphism to
$H$ if and only if the complement of $G$ is 3-colourable.) The trigraph in
Figure \ref{fig:ex1}c illustrates another way to have this take place. Here, the
two other vertices have weak loops and are joined by a strong symmetric edge. In
this case L-HOM($H$) is also NP-complete, since it corresponds (in the
complement) to the stable cutset problem \cite{sula}.  Thus we again have an
easy \mbox{L-HOM($H^-$)} with a hard L-HOM($H$).

Of course, it is possible that some class of trigraphs enjoys dichotomy for
reasons different from polynomial equivalence of L-HOM($H$) and L-HOM($H^-$).
However, the stubborn problem (Figure \ref{fig:ex1}a) illustrates the fact that
there are trigraphs $H \not\in {\mathcal T}$ where even the dichotomy is not
clear. We note that the trigraph $H$ for the stubborn problem (the trigraph in
Figure \ref{fig:ex1}a) violates the domination property, since the statement
(D2) does not hold for the strong loop at $x$.

Finally, we have also introduced the surjective list homomorphism problem
SL-HOM($H$) as an interesting variant of L-HOM($H$). We were able to completely
classify the complexity of SL-HOM($H$) for trigraphs $H$ in the class $\mathcal
S$ by proving polynomial equivalence with SL-HOM($H^{--})$.  This result
implies, in particular, that the two problems, L-HOM($H$) and SL-HOM($H$), may
not necessarily have the same complexity for all trigraphs $H$.  This difference
was not noted before, since the two problems are polynomially equivalent for all
digraphs $H$. In fact, we have given a general condition for trigraphs $H$ under
which the two problems SL-HOM($H$) and L-HOM($H$) are polynomially equivalent.
Nonetheless, the trigraph $H$ in Figure~\ref{fig:ex1}d illustrates a case where
SL-HOM($H$) is polynomial time solvable, because using representatives and
arc-consistency reduces all lists to size at most two, while L-HOM($H$) is
$NP$-complete, because $H-x$ corresponds to the stable cutset problem
\cite{sula}.

\end{document}